\documentclass[5p,times]{elsarticle}

\journal{Computer-Aided Design}
\usepackage{amssymb}

\newtheorem{theorem}{Theorem}
\newproof{proof}{Proof}
  
\begin{document}

\begin{frontmatter}

\title{Geometric Rounding and Feature Separation in Meshes}

\author{Victor Milenkovic}
\ead{vjm@cs.miami.edu}
\address{Department of Computer Science, University of Miami,
Coral Gables, FL 33124-4245, USA}

\author{Elisha Sacks\corref{cor1}}
\ead{eps@purdue.edu}
\address{Computer Science Department, Purdue University,
  West Lafayette, IN 47907-2066, USA}

\cortext[cor1]{Correponding author}

\begin{abstract}
Geometric rounding of a mesh is the task of approximating its vertex coordinates
by floating point numbers while preserving mesh structure.  Geometric rounding
allows algorithms of computational geometry to interface with numerical
algorithms.  We present a practical geometric rounding algorithm for 3D triangle
meshes that preserves the topology of the mesh.  The basis of the algorithm is a
novel strategy: 1) modify the mesh to achieve a feature separation that prevents
topology changes when the coordinates change by the rounding unit; and 2) round
each vertex coordinate to the closest floating point number.  Feature separation
is also useful on its own, for example for satisfying minimum separation rules
in CAD models.  We demonstrate a robust, accurate implementation.
\end{abstract}


\end{frontmatter}

\section{Introduction}

A common representation for a surface is a triangle mesh: a set of disjoint
triangles with shared vertices and edges.  Meshes are usually constructed from
basic elements (polyhedra, triangulated surfaces) through sequences of
operations (linear transformations, Booleans, offsets, sweeps).  Although the
vertices of the basic elements have floating point coordinates, the mesh
vertices can have much higher precision.  For example, the intersection point of
three triangles has thirteen times the precision of their vertices.  High
precision coordinates are incompatible with numerical codes that use floating
point arithmetic, such as finite element solvers.  Rewriting the software to use
extended precision arithmetic would be a huge effort and would entail an
unacceptable performance penalty.  Instead, the coordinates must be approximated
by floating point numbers.

One strategy is to construct meshes using floating point arithmetic, so the
coordinates are rounded as they are computed.  This strategy suffers from the
robustness problem: software failure or invalid output due to numerical error.
The problem arises because even a tiny numerical error can make a geometric
predicate have the incorrect sign, which in turn can invalidate the entire
algorithm.  For a controlled class of inputs, robustness can be ensured by
careful engineering of the software, but the problem tends to recur when new
domains are explored.

A solution to the robustness problem, called Exact Computational Geometry (ECG)
\cite{ecg}, is to represent geometry exactly and to evaluate predicates exactly.
ECG is efficient because most predicates can be evaluated exactly using floating
point arithmetic.  Mesh construction using ECG is common in computational
geometry research and plays a growing role in applications.  We use ECG to
implement Booleans, linear transformations, offsets, and Minkowski sums
\cite{sacks-milenkovic13b,kyung-sacks-milenkovic15}; the CGAL library
\cite{cgal} provides many ECG implementations.

Although ECG guarantees a correct mesh with exact vertex coordinates, it does
not provide a way to approximate the coordinates in floating point.  Replacing
the coordinates with the nearest floating point numbers can cause triangles to
intersect.  The task of constructing an acceptable approximation is called
\emph{geometric rounding}.  Prior work solves the problem in 2D
\cite{hp-isr-02,m-spr97,goodrich97snap}, but the 3D problem is open
(Sec.~\ref{s-prior}).

\paragraph*{Contribution}

We present a practical geometric rounding algorithm for 3D triangle meshes.  We
separate the close features then round the vertex coordinates.  The separation
step ensures that the rounding step preserves the topology of the mesh.

The features of a mesh are its vertices, edges, and triangles.  Two features are
disjoint if they share no vertices.  The mesh is $d$-separated if the distance
between every pair of disjoint features exceeds $d$.  Moving each vertex of a
$d$-separated mesh by a distance of at most $d/2$ cannot make any triangles
intersect.  If the vertex coordinates are bounded by $M$, rounding them to the
nearest floating point numbers moves a vertex at most $e=\sqrt{3}M\epsilon$ with
$\epsilon$ the rounding unit.  Thus, $2e$-separation followed by rounding
prevents triangle intersection, which in turn prevents changes in the topology
of the mesh.  In particular, channels cannot collapse and cells cannot merge.

We separate a mesh in three stages.  Modification removes short edges and skinny
triangles by means of mesh edits (Sec.~\ref{s-step1}).  Expansion iterates
vertex displacements that maximize the feature separation (Sec.~\ref{s-step2}).
Optimization takes a separated mesh as input and minimizes the total
displacement of the vertices from their original input values under the
separation constraints (Sec.~\ref{s-step3}).  Modification quickly separates
most close features, thereby accelerating the other stages.  Expansion and
optimization guarantee separation via a locally minimum displacement.

The mesh separation algorithm builds on prior work in mesh simplification and
improvement (Sec.~\ref{s-prior}).  Modification performs one type of
simplification and the other stages perform one type of improvement.  We provide
a novel form of mesh improvement that is useful in applications that suffer from
close features.  For example, we can remove ill-conditioned triangles from
finite element meshes.  Another important application of mesh separation is
computer-aided design, which typically requires a separation of $10^{-6}$ due to
manufacturing constraints.  Current software removes small features
heuristically, e.g. by merging close vertices, which can create triangle
intersections.  Our algorithm provides a safe, efficient alternative.

\paragraph*{Implementation}

We implement the geometric rounding algorithm robustly (Sec.~\ref{s-imp}).  We
use ECG for the computational geometry tasks, such as finding close features and
testing if triangles intersect.  We implement expansion and optimization via
linear programming, using adaptive scaling to solve accurately and efficiently.
We demonstrate the implementation on meshes of varying type, size, and vertex
precision (Sec.~\ref{s-results}).  We conclude with a discussion
(Sec.~\ref{s-con}).

\section{Prior work}\label{s-prior}

We survey prior work related to geometric rounding of 3D meshes.  Fortune
\cite{fortune99} rounds a mesh of size $n$ in a manner that increases the
complexity to $n^4$ in theory and in practice.  Fortune \cite{fortune97} gives a
practical rounding algorithm for plane-based polyhedra.  We
\cite{sacks-milenkovic13b} extend the algorithm to polyhedra in a mesh
representation.  Most output vertices have floating point coordinates, and the
highest precision output vertex is an intersection point of three triangles with
floating point coordinates.  Zhou et al \cite{zhou2016} present a heuristic form
of geometric rounding akin to our algorithm and show that it works on 99.95\% of
10,000 test meshes.  This paper improves on our prior algorithm in three ways:
the input can be any mesh, all the output vertices have floating point
coordinates, and the mesh topology is preserved.

There is extensive research \cite{cignoni98} on mesh simplification: accurately
approximating a mesh of small triangles by a smaller mesh of larger triangles.
The mesh edits in the modification stage of our algorithm come from this work.
Simplification reduces the number of close features as a side effect of reducing
the number of triangles and increasing their size, but removing all the close
features is not attempted.  Mesh untangling and improvement algorithms
\cite{freitag2000,knupp2001} improve a mesh by computing vertex displacements
via optimization.  Although we use a similar strategy, our objective functions
and optimization algorithms are novel.  Cheng, Dey, and Shewchuk \cite{cheng12}
improve Delaunay meshes with algorithms that remove some close features, notably
sliver tetrahedra.

\section{Modification}\label{s-step1}

The modification stage of the feature separation algorithm consists of a series
of mesh edits: edge contractions remove short edges and edge flips remove skinny
triangles.  These edits are used in 2D geometric rounding, in mesh
simplification, and elsewhere.  We perform every edit that satisfies its
preconditions and that does not create triangle intersections.  This strategy
separates most close features in a wide range of meshes (Sec.~\ref{s-results}).

An edge $th$ is short if $||h-t||<d$.  For $th$ to be contracted, it must be
incident on two triangles, $thv$ and $htw$ (Fig.~\ref{f-vv}a).  The triangles
that contain $t$ on their boundary form a surface whose boundary is a simple
loop $h,v,t_1,\ldots,t_m,w$; likewise for $h$ with $t,w,h_1,\ldots,h_n,v$.  If
the $t_i$ and the $h_j$ are disjoint sets, $th$ is contracted: it is replaced
with a new vertex $m=(t+h)/2$ and the incident edges and triangles are updated
(Fig.~\ref{f-vv}b).

\begin{figure}[htbp]
  \begin{tabular}{cc}
    \includegraphics[scale=.75]{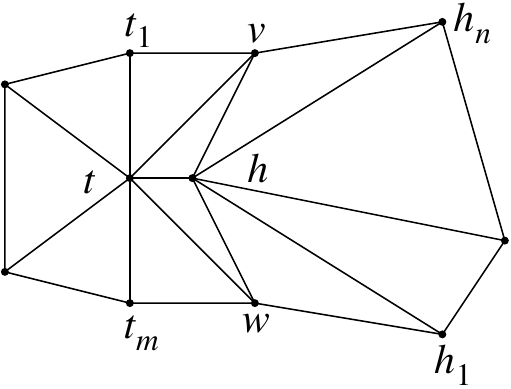} & \includegraphics[scale=.75]{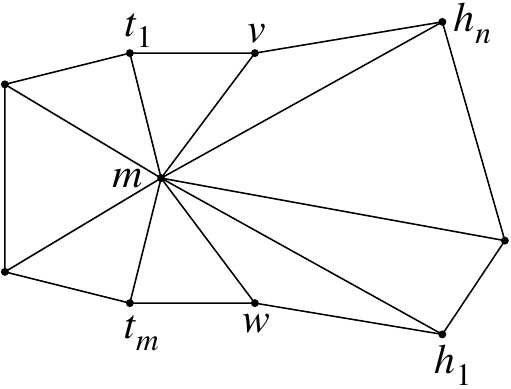}\\
    (a) & (b)
  \end{tabular}
  \caption{(a) Short edge $th$; (b) $th$ contraction.}\label{f-vv}
\end{figure}

A triangle $thv$ is skinny if $v$ projects onto a point $p$ in $th$ and
$||p-v||<d$ (Fig.~\ref{f-ve}a).  The edge $th$ is flipped if it is incident on
two triangles, $thv$ and $htw$, $vw$ is not an edge of the mesh, and the
triangles $vwh$ and $wvt$ are not skinny.  The flip replaces $th$ with $vw$, and
replaces $thv$ and $htw$ with $vwh$ and $wvt$ (Fig.~\ref{f-ve}b).

\begin{figure}[htbp]
  \begin{tabular}{cc}
    \includegraphics[scale=.75]{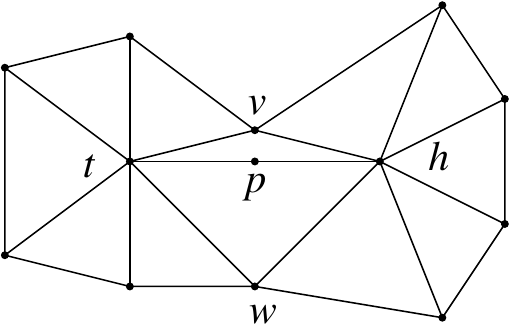} & \includegraphics[scale=.75]{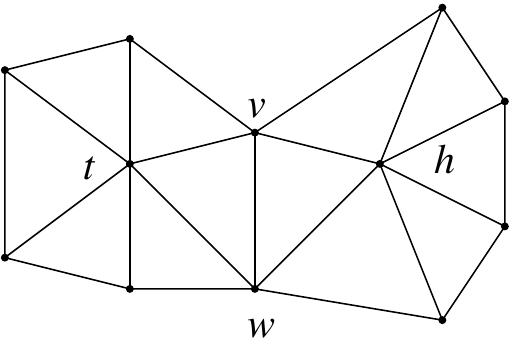}\\
    (a) & (b)
  \end{tabular}
  \caption{(a) Skinny triangle $thv$; (b) $th$ flip.}\label{f-ve}
\end{figure}

\subsection*{Analysis}

Modification terminates because every contraction reduces the number of edges
and every flip maintains this number and reduces the number of skinny triangles.
The computational complexity of an edit is constant, except for triangle
intersection testing, which is linear in the number of mesh triangles.  We
cannot bound the number of edits in terms of the number of short edges and
skinny triangles in the input because edge contractions can create skinny
triangles.

Edits preserve the intrinsic topology of the mesh because they replace a
topological disk by another topological disk with the same boundary.  The
boundary is $t_1,\ldots,t_m,w,h_1,\ldots,h_n,v$ for a short edge and is $vtwh$
for a skinny triangle.  This property and the rejection of edits that cause
triangle intersections jointly preserve the extrinsic topology of the connected
components of the mesh.  Components whose volume or thickness is less than $d$
can change their nesting order.  We remove these components because applications
find them useless at best.

We measure error by the distance between the input and the output meshes.  An
edit deforms its topological disk by at most $d$, so the error due to
modification is bounded by $d$ times the number of edits.  Although we lack a
bound in terms of the input, the median error is at most $d$ in our tests
(Sec.~\ref{s-results}).

\section{Expansion}\label{s-step2}

The expansion stage of the feature separation algorithm iteratively increases
the separation until the mesh is $d$-separated.  Each iteration assigns every
vertex $a$ a displacement $a'$ and a new position $a+a'$ that maximize the mesh
separation to first order.  The coordinate displacements are bounded by
$\Delta=d$ to control the truncation error.  The iterator verifies that the true
separation of the mesh increases and that the topology is preserved.  If not,
its halves $\Delta$ and tries again.

The mesh separation is the minimum distance between a vertex and a triangle or
between two edges.  The distance between features $A$ and $B$ is the maximum
over unit vectors $u$ of the minimum of $u\cdot(b-a)$ over the vertices $a\in A$
and $b\in B$.  The optimal $u$ is parallel to $q-p$ with $p$ and $q$ the closest
points of $A$ and $B$ (Fig.~\ref{f-close}).  We approximate the distance between
the displaced features by a linear function of the displacements.  The first
order displacement of $u$ is expressible as $lv+mw$ with $u$, $v$, and $w$
orthonormal vectors.  We obtain the first order distance
\begin{equation}\label{e-d}
\min_{a\in A, b\in B} u\cdot(b-a)+u\cdot(b'-a')+(lv+mw)\cdot(b-a)
\end{equation}
by substituting the displaced values into $u\cdot(b-a)$ and dropping the
quadratic terms.  Although $l$ and $m$ are determined by the vertex
displacements in the standard first order distance formula, we treat them as
optimization variables.  We explain this decision below (Sec.~\ref{s-con}).

\begin{figure}[htbp]
  \begin{tabular}{cc}
    \includegraphics{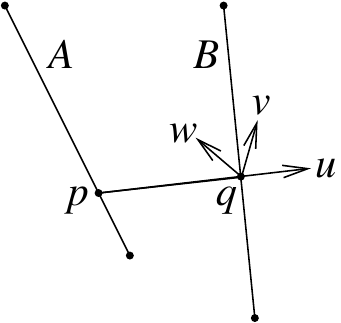} & \includegraphics{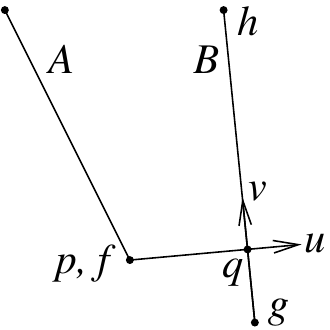}\\
    (a) & (b)
  \end{tabular}
  \caption{Close features $A$ and $B$ with closest points $p$ and $q$: (a) two
    edges, (b) vertex and triangle.} \label{f-close}
\end{figure}

We compute the displacements by solving two linear programs (LPs).  The first LP
maximizes the first order distance between the closest pair of features subject
to the $\Delta$ bounds.  It constrains the first order distance between every
pair to exceed $sd$ with $0\leq s\leq1$ a variable that it maximimizes.  Since
$s=1$ implies $d$-separation, a larger value would increase the error for no
reason.  The second LP computes minimal displacements that achieve the maximal
$s$ value.

We need to constrain the features whose distance is less than $d$, so the LPs
will $d$-separate them, and the features whose distance is slightly greater, so
the LPs will not undo their $d$-separation or make them intersect.  We achieve
these goals, while avoiding unnecessary constraints, by constraining the
features whose distance is less than $2\sqrt{3}d$.  The LPs cannot cause an
intersection between an unconstrained pair because of the $\Delta$ bounds.  They
can undo its $d$-separation, perhaps causing an extra iteration, but this never
happens in our tests.

The two LPs appear below.  The position of a vertex $a$ is a constant, also
denoted by $a$, and its displacement $a'$ defines three variables $a'_{x,y,z}$
(shorthand for $a'_x,a'_y,a'_z$).  Other variables and constants are defined
above.  The first LP computes a maximal $s=s_\mathrm{m}$.  The second LP
computes minimal displacements that achieve $s=s_\mathrm{m}$.  The magnitude of
the displacement of a vertex $a$ is represented by variables
$a^{\mathrm{m}}_{x,y,z}$ with constraints $-a^{\mathrm{m}}_{x,y,z}\leq
a'_{x,y,z} \leq a^{\mathrm{m}}_{x,y,z}$.  The objective is to minimize the sum
over the vertices of these variables.  The variable $s$ is replaced by the
constant $s_\mathrm{m}$ in the constraints.

\begin{flushleft}
\begin{minipage}{\columnwidth}
\centerline{First Expansion LP}
\textbf{Constants:} $d$, $\Delta$, the vertices, and the vectors $u$, $v$,
and $w$ for each pair of features.

\textbf{Variables:} $s$; $a'_{x,y,z}$ for each vertex $a$; $l$ and $m$ for
each pair of features.

\textbf{Objective:} maximize $s$.

\textbf{Constraints:} $0\leq s\leq 1$, $-\Delta\leq a'_{x,y,z}\leq\Delta$, and
$\mathrm{Eq.~(\ref{e-d})}\geq sd$.
\end{minipage}
\end{flushleft}

\begin{flushleft}
\begin{minipage}{\columnwidth}
\centerline{Second Expansion LP}
\textbf{Constants:} $s_m$, $d$, $\Delta$, the vertices, the vectors $u$,
$v$, and $w$ for each pair of features. 

\textbf{Variables:} $a'_{x,y,z}$ and $a^m_{x,y,z}$ for each vertex $a$; $l$
and $m$ for each pair of features.

\textbf{Objective:} minimize $\sum_a a^m_x+a^m_y+a^m_z$.

\textbf{Constraints:} $-\Delta\leq a'_{x,y,z}\leq\Delta$,
$-a^{\mathrm{m}}_{x,y,z}\leq a'_{x,y,z} \leq a^{\mathrm{m}}_{x,y,z}$, and
$\mathrm{Eq.~(\ref{e-d})}\geq d$.
\end{minipage}
\end{flushleft}

\subsection*{Analysis}

We prove that expansion terminates.  The first LP computes an $s>0$ because no
triangles intersect.  The displacements that it computes are feasible for the
second LP.  Hence, every expansion step succeeds for some $\Delta>0$.  It
remains to bound the number of steps.  Although we only prove a weak bound, the
number of steps appears to be a small constant (usually 1) based on extensive
testing (Sec.~\ref{s-results}).

We employ the following definitions.  A $\Delta$-displacement is a displacement
in which the coordinate displacements are bounded by $\Delta$.  Let $\delta$
denote the minimum separation of the mesh.  The \emph{tightness} $\tau$ as the
minimum over all $\Delta$-displacements of the ratio of $\Delta$ to the increase
in $\delta$.

\begin{theorem}
  The number of expansion steps is bounded.
\end{theorem}

\begin{proof}
A $\Delta$-displacement can increase the separation of a pair by at most
$2\sqrt{3}\Delta$, so $\tau>\sqrt{3}/6$.  Let $M$ be the maximum vertex
coordinate magnitude.  The $\Delta$-displacement $a'=\Delta a/M$ increases
$\delta$ by $\Delta\delta/M$, so $\tau<\Delta/(\Delta\delta/M)=M/\delta$.  We
assume that a displacement that increases the separation of the mesh does not
increase $M/\delta$; it does not decrease $\delta$ by definition.  Hence,
$\tau<M/\delta$ at every expansion step.

The LP maximizes $S=sd$ for the linear separation constraints.  Let $S^\ast$ be
the maximum of $S$ for the true constraints and let $T$ be the maximum
truncation error of the linear constraints.  The linear separations for the
optimal displacement are at least $S^\ast-T$.  Therefore, the LP solution
satisfies $S\geq S^\ast-T$ and the true separation is at least $S-T\geq
S^\ast-2T$.

For a $\Delta$-displacement with $\Delta=\delta/k$, $T<c\delta/k^2$ for some
constant $c$.  For $k=4c\tau$, the optimal $\Delta$-step increases $\delta$ by
$\delta/k\tau=\delta/4c\tau^2$, so the LP increases $\delta$ by
$\delta/4c\tau^2-2c\delta/k^2=\delta/8c\tau^2$.  Since one step multiplies
$\delta$ by $1+1/8c\tau^2$, ${\mathrm{O}}(\tau^2\log(d/\delta))$ steps increase
it to $d$.
\end{proof}

\section{Optimization}\label{s-step3}

The optimization stage of the feature separation algorithm takes the output of
the expansion stage as input and locally minimizes the displacement of the
vertices from their original positions.  (Expansion does not minimize the
displacement even though each of its iterations is approximately optimal.)  We
apply a gradient descent strategy to the semi-algebraic set of vertex values for
which the mesh is $d$-separated.  We compute the direction in which the
displacement decreases most rapidly subject to the linear separation
constraints.  We take a step of size $\beta=d$ in this direction and use the
expansion algorithm to correct for the linearization error and return to the
$d$-separated space.  If the displacement increases, we divide $\beta$ by 2 and
retry the step. If four consecutive steps succeed, we double $\beta$.

We compute the descent direction by solving an LP.  We represent the
pre-expansion value of a vertex $a$ by constants $a^0_{x,y,z}$ and represent its
displacement by variables $a^{\mathrm{d}}_{x,y,z}$ with constraints
$-a^{\mathrm{d}}_{x,y,z}\leq a_{x,y,z}+a'_{x,y,z}-a^0_{x,y,z}\leq
a^{\mathrm{d}}_{x,y,z}$ and $-\beta\leq a'_{x,y,z}\leq\beta$.  The objective is
to minimize the sum over the vertices of $a^{\mathrm{d}}_{x,y,z}$.  We drop $s$
from the constraints because they are feasible with $s=1$ due to expansion.

\begin{flushleft}
\begin{minipage}{\columnwidth}
\centerline{Optimization LP}
\textbf{Constants:} $\beta$, $d$, the vertices, the pre-expansion vertices
$a^0$, and the vectors $u$, $v$, and $w$.

\textbf{Variables:} the $a'_{x,y,z}$, $l$, $m$, and $a^d_{x,y,z}$.

\textbf{Objective:} minimize $\sum_a a^d_x+a^d_y+a^d_z$.

\textbf{Constraints:} $-\beta\leq a'_{x,y,z}\leq\beta$,
$-a^{\mathrm{d}}_{x,y,z}\leq a_{x,y,z}+a'_{x,y,z}-a^0_{x,y,z}\leq a^{\mathrm{d}}_{x,y,z}$,
and $\mathrm{Eq.~(\ref{e-d})}\geq d$.
\end{minipage}
\end{flushleft}

\section{Implementation}\label{s-imp}

\paragraph*{Modification}

We accelerate the intersection tests by storing the mesh in an octree.  We
perform mesh edits in an order that reduces the number of tests.  Contractions
come before flips because removing a short edge also eliminates two skinny
triangles.  We perform edits of the same type in order of feature distance.
When we contract an edge of length $x$, we need only test for intersections
between the new triangles and old triangles within distance $x/2$.  When we
remove a skinny triangle of size $x$, we need only test for intersections
between the two new triangles and old triangles within distance $x$.  In both
cases, prior edits have eliminated most such old triangles.

\paragraph*{Expansion}

We combine the two LPs into one LP that achieves similar results in half the
time.  We take the constants, variables, and constraints from the first LP and
add the $a^m$ and their constraints.  We maximize $s-bn\sum_a a^m_x+a^m_y+a^m_z$
with $b$ a large constant and with $n$ the number of vertices.

\paragraph*{Linear programs}

We use the IBM CPLEX solver.  The input size is not a concern because the number
of constraints is proportional to the number of close features.  The challenge
is to formulate a well-conditioned LP.  We scale the displacement variables by
$d$ because their values are small multiples of $d$.  We divide the separation
constraints by $d$ to avoid tiny coefficients.  We bound the magnitudes of $l$
and $m$ by 0.001 and scale them by a parameter $\alpha$ that is initialized to
$1$.  After solving an LP, we check if the solution violates any of the
constraints by over $10^{-6}$, meaning that the separation is off by over
$10^{-6}d$.  If so, we multiply $\alpha$ by 10 and resolve the LP.  This
strategy trades off accuracy for number of iterations when scaling is poor.

We test every constrained pair of features $A$ and $B$ for intersection.  If
$(u+lv+mw)\cdot(b+b'-a-a')>0$ for all vertices $a\in A$ and $b\in B$, the pair
does not intersect.  Otherwise, we express the displacement of each vertex $v$
as a linear transform $v+v't$ and test if $A$ and $B$ are in contact at any
$t\in[0,1]$.  The vertices of $A$ and $B$ must be collinear, which occurs at the
zeros of a cubic polynomial, and the transformed features must intersect.  The
conservative test, is much faster than the exact test and rarely rejects a good
step.

\section{Results}\label{s-results}

We validated the feature separation algorithm on five types of meshes: open
surfaces, three types of closed manifolds with increasing vertex complexity, and
tetrahedral meshes.  We describe the results for $d=10^{-6}$ because this is the
typical minimum feature size in CAD software.  Decreasing $d$ decreases the
number of close features and hence the running time roughly proportionally.  The
running times are for one core of an Intel Core i7-6700 CPU at 3.40GHz with 16
GB of RAM.

\subsection{Isosurfaces}

The inputs are isosurfaces with 53~bit vertex coordinates that are constructed
by a marching cubes algorithm, and simplified isosurfaces (Fig.~\ref{f-mesh}).
Table~\ref{t-simp} shows the results for four isosurfaces with 35,000 to 125,000
triangles and for two simplified isosurfaces with 1,500 and 2,500 triangles.
Feature separation modifies 0.0\% of the vertices of the isosurfaces; it
modifies median 1.4\% of the vertices of the simplified isosurfaces with median
error $0.3d$.  The median and maximum running times per input triangle are 4e-5
and 0.002 seconds.

\begin{figure}[htbp]
  \begin{tabular}{cc}
    \includegraphics[scale=.165]{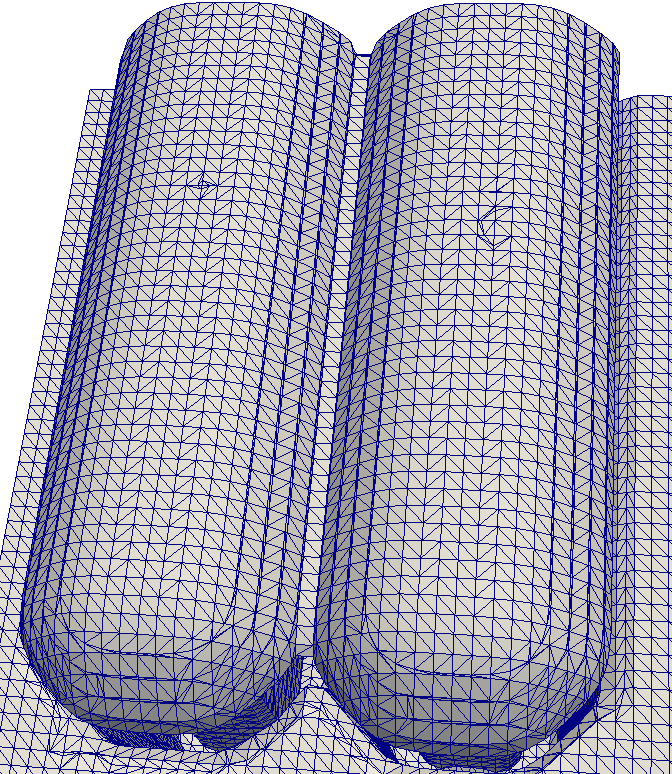} & \includegraphics[scale=.165]{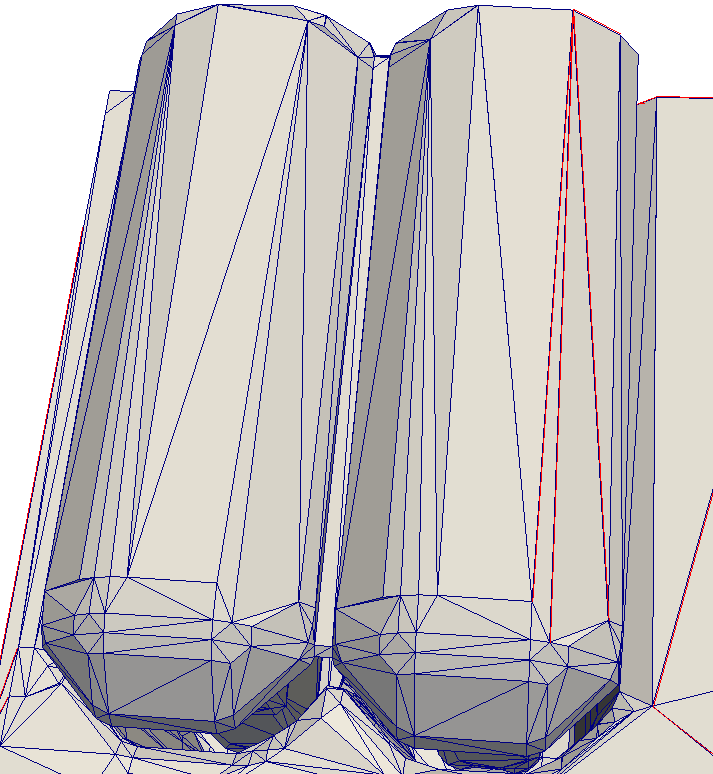}\\
    (a) & (b)
  \end{tabular}
  \caption{(a) Isosurface; (b) simplified mesh with close features in
    red.} \label{f-mesh}
\end{figure}

\begin{table*}[tbp]
  \caption{Simplification: $f$ triangles, $c_m$ $d$-close features, $v_m$
    percentage vertices displaced, $a_m$ median and $m_m$ max displacement in
    units of $d$ for modification, likewise $c_e,v_e,a_e,m_e$ for expansion and
    $v_o,a_o,m_o$ for optimization, and $t$ seconds running time.}\label{t-simp}
  \begin{tabular}{rccccccccccccc}
    sec. & $f$ & $c_m$ & $v_m$ & $a_m$ & $m_m$ & $c_e$ & $v_e$ & $a_e$ & $m_e$ & $v_o$ &
    $a_o$ & $m_o$ & $t$\\
    7.1 & 34876 & 36 & 0.00 & 0.00 & 0.00 & 8 & 0.02 & 0.00 & 0.00 & 0.02 & 0.00
    & 0.00 & 1.23\\ 
    & 35898 & 36 & 0.00 & 0.00 & 0.00 & 8 & 0.02 & 0.00 & 0.00 & 0.02 & 0.00 &
    0.00 & 1.26\\ 
    & 123534 & 2 & 0.00 & 0.00 & 0.00 & 0 & 0.00 & 0.00 & 0.00 & 0.00 & 0.00 &
    0.00 & 2.36\\ 
    & 125394 & 5 & 0.00 & 0.00 & 0.00 & 0 & 0.00 & 0.00 & 0.00 & 0.00 & 0.00 &
    0.00 & 2.37\\ 
    & 1596 & 139 & 0.37 & 0.56 & 0.65 & 11 & 0.50 & 0.18 & 0.30 & 0.50 & 0.17 &
    0.25 & 0.45\\ 
    & 2568 & 177 & 0.07 & 0.13 & 0.13 & 66 & 2.41 & 0.85 & 3.00 & 2.41 & 0.48 &
    1.25 & 4.33\\\hline 
    7.2 & 377534 & 5119 & 0.13 & 0.01 & 0.70 & 69 & 0.01 & 0.60 & 2.14 & 0.01 &
    0.46 & 2.24 & 15.58\\ 
    & 407014 & 18727 & 0.54 & 0.02 & 0.71 & 149 & 0.02 & 0.36 & 1.71 & 0.02 &
    0.32 & 1.45 & 22.45\\ 
    & 436308 & 23980 & 0.23 & 0.33 & 1.02 & 441 & 0.07 & 0.48 & 2.21 & 0.07 &
    0.41 & 1.57 & 211.00\\ 
    & 438364 & 291 & 0.00 & 0.40 & 0.83 & 11 & 0.00 & 0.36 & 0.73 & 0.00 & 0.36
    & 0.72 & 8.60\\ 
    & 641604 & 221 & 0.00 & 0.53 & 0.70 & 14 & 0.00 & 0.48 & 0.89 & 0.00 & 0.46
    & 0.86 & 11.83\\ 
    & 667838 & 329 & 0.00 & 0.54 & 0.94 & 13 & 0.00 & 0.51 & 1.03 & 0.00 & 0.47
    & 0.83 & 10.06\\ 
    & 773236 & 235 & 0.00 & 0.30 & 0.48 & 7 & 0.00 & 0.43 & 1.00 & 0.00 & 0.36 &
    0.98 & 14.62\\\hline 
    7.3 & 1572 & 1916 & 5.69 & 0.07 & 0.55 & 120 & 1.89 & 0.85 & 2.45 & 3.67 &
    0.74 & 2.24 & 5.32\\ 
    & 1612 & 2614 & 6.91 & 0.08 & 0.74 & 81 & 1.85 & 0.85 & 2.45 & 3.95 & 0.68 &
    1.98 & 3.83\\\hline
    7.4 & 465546 & 1363142 & 18.54 & 0.08 & 1.04 & 9250 & 1.44 & 0.85 & 3.01 &
    0.00 & 0.00 & 0.00 & 299.90\\ 
    & 473542 & 1506798 & 14.39 & 0.08 & 1.31 & 12186 & 2.34 & 0.59 & 3.05 & 0.00
    & 0.00 & 0.00 & 459.50\\\hline 
    7.5 & 125986 & 4 & 0.00 & 0.00 & 0.00 & 4 & 0.07 & 0.83 & 1.00 & 0.07 & 0.83
    & 2.83 & 14.59\\ 
    & 154084 & 0 & 0.00 & 0.00 & 0.00 & 0 & 0.00 & 0.00 & 0.00 & 0.00 & 0.00 &
    0.00 & 3.03
  \end{tabular}
\end{table*}

\subsection{Minkowski sums}

The inputs are Minkowski sums of pairs of polyhedra with 53~bit vertex
coordinates (Fig.~\ref{f-mink}).  We construct the Minkowski sums robustly using
our prior algorithm \cite{kyung-sacks-milenkovic15}, but without perturbing the
input.  The output vertex coordinates are ratios of degree~7 and degree~6
polynomials in the input coordinates, hence have about 800~bit precision.

\begin{figure}[htbp]
  \begin{tabular}{cc}
    \includegraphics[scale=.2]{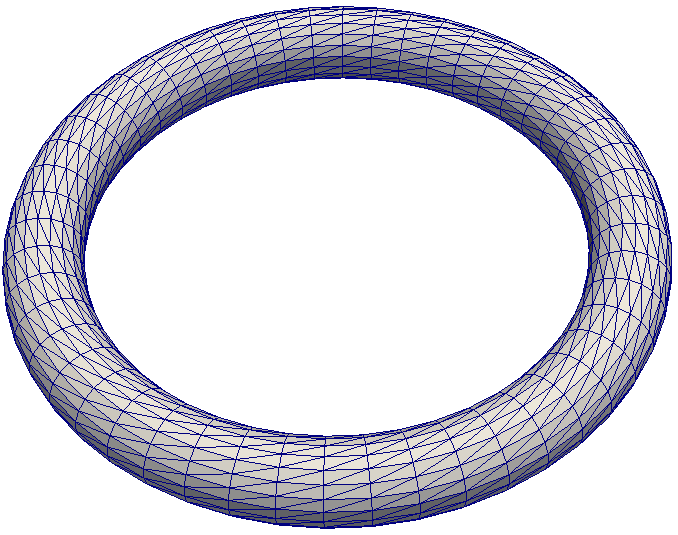} & \includegraphics[scale=.2]{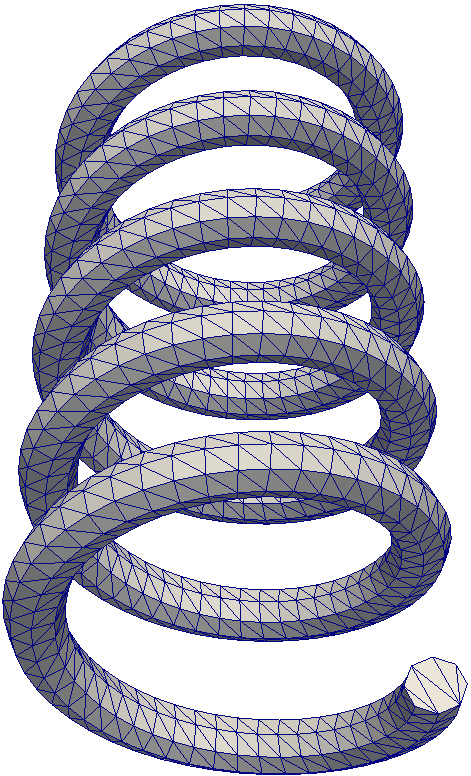}\\
    (a) & (b)\\
    \multicolumn{2}{c}{\includegraphics[scale=.2]{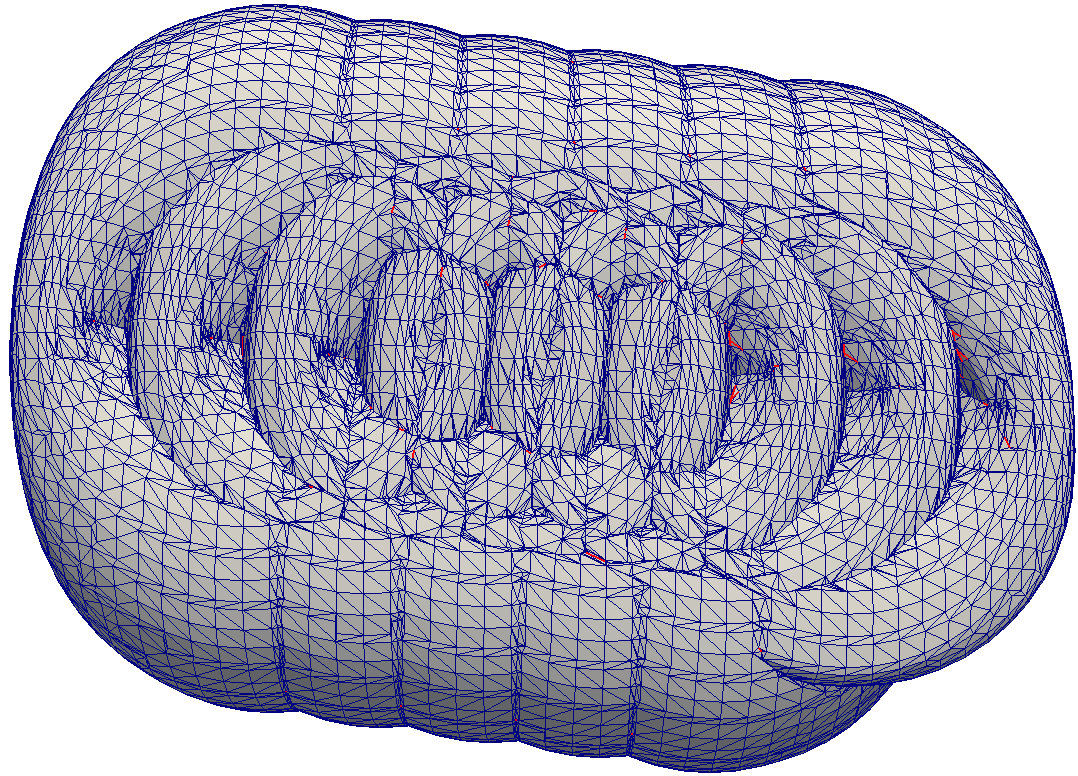}}\\
    (c)
  \end{tabular}
  \caption{(a) Torus, (b) helix, (c) Minkowski sum with close features in
    red.}\label{f-mink}
\end{figure}

We test 10 polyhedra with up to 40,000 triangles.  The 55 pairs have Minkowski
sums with median 32,000 and maximum 773,000 triangles, and median 80 and maximum
37,000 close features.  Modification reduces the number of close features to
median 0 and maximum 440, and displaces median 0.01\% and maximum 0.58\% of the
vertices with median and maximum displacements of $0.15d$ and $d$.  Expansion
displaces median 0.0\% and maximum 0.07\% of the vertices with median and
maximum displacements of $0.0d$ and $5.33d$.  Optimization does not help.  The
median and maximum running times per input triangle are $2\times10^{-5}$ and
$0.001$ seconds.  Table~\ref{t-simp} shows the results for the seven largest
inputs.

\subsection{Sweeps}

The inputs are approximate 4D free spaces for a polyhedral robot that rotates
around the $z$ axis and translates freely relative to a polyhedral obstacle.
The algorithm is as follows.  We split the rotation into short intervals.
Within an interval, we approximate the free space by the Minkowski sum of the
obstacle with the complement of the volume swept by the robot as it rotates
around the axis.  We employ a rational parameterization of the rotation matrix
with parameter $t$.  The coordinates of a rotated vertex are ratios of
polynomials that are quadratic in $t$ and linear in the input coordinates, for
total degree 6.  The highest precision vertices of the swept volume are
intersection points of three triangles comprised of rotated vertices.  The
degree of their coordinates is $13\times6=78$ and the input has 27~bit
precision, so the output precision is about 2100 bits.

We test a robot with 12 triangles and an obstacle with 64 triangles
(Fig.~\ref{f-drone}).  We use 40 equal rotation angle intervals.  The inputs
have 1,400 to 1,600 triangles, and median 1,600 maximum 3,600 close features.
Modification reduces the number of close features to median 70 maximum 140, and
displaces median 5.2\% and maximum 8.5\% of the vertices with median and maximum
displacements of $0.08d$ and $0.87d$.  Expansion displaces median 2.0\% and
maximum 2.2\% of the vertices with median and maximum displacements of $0.86d$
and $2.46d$, which optimization changes to $0.57d$ and $6.1d$.  The median and
maximum running times per input triangle are 0.004 and 0.02 seconds.
Table~\ref{t-simp} shows the results for the two largest inputs.

\begin{figure}[htbp]
  \begin{tabular}{cc}
    \includegraphics[scale=.25]{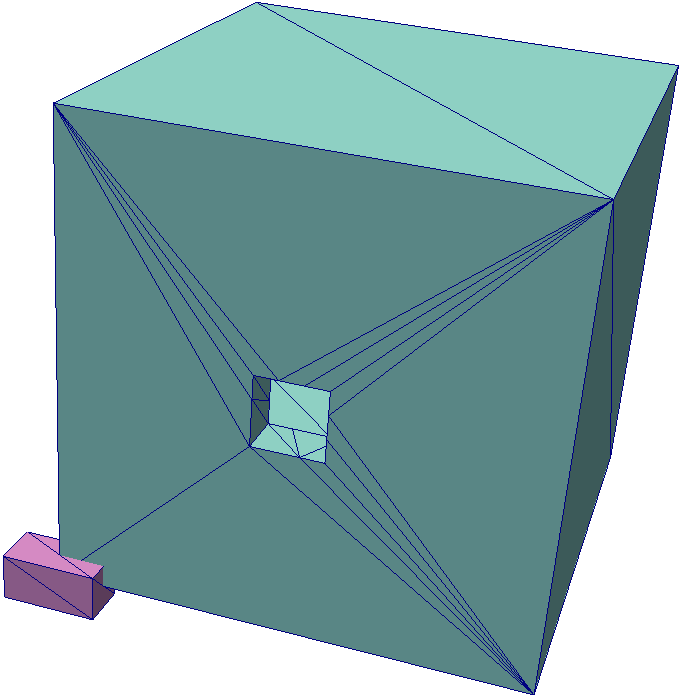} \\
    (a)\\
    \includegraphics[scale=.18]{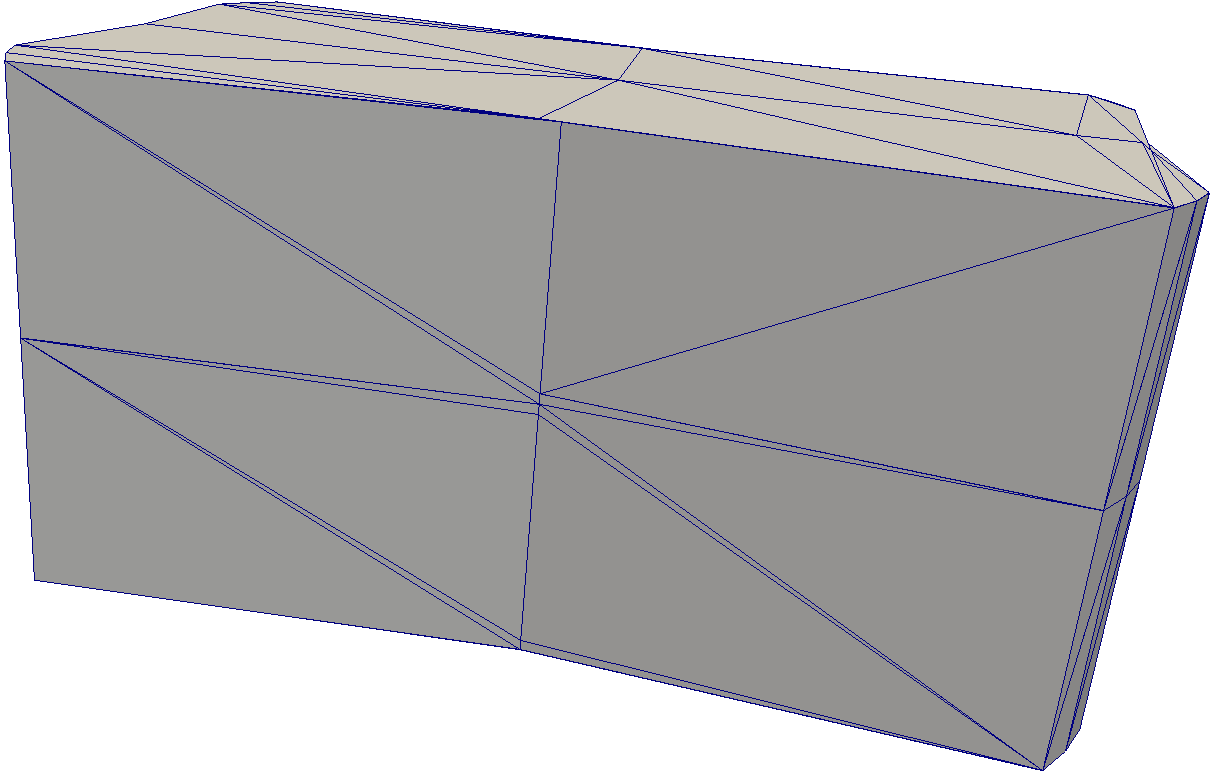}\\
    (b)\\
    \includegraphics[scale=.25]{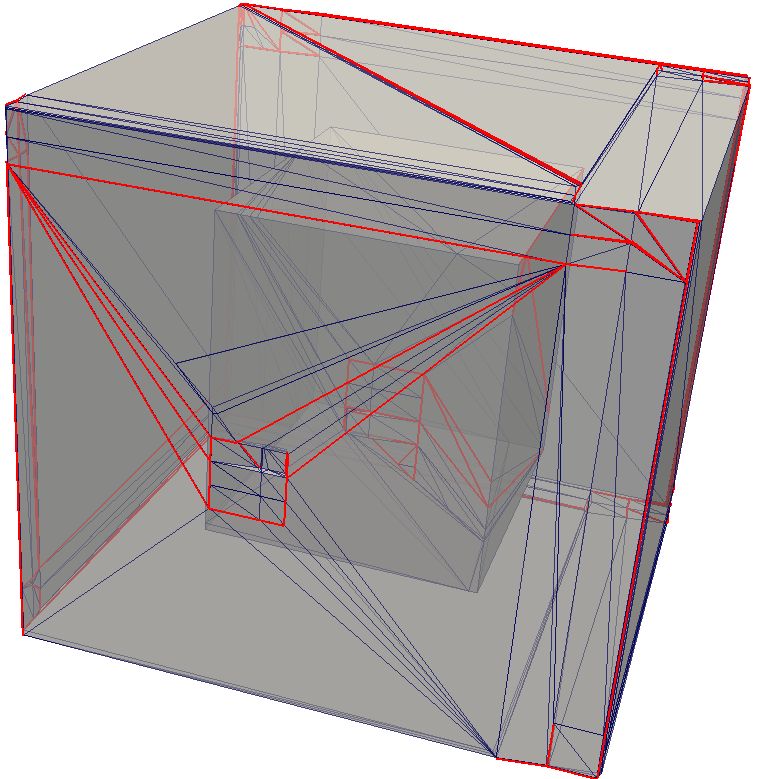}\\
    (c)
  \end{tabular}
  \caption{(a) Robot (purple) and obstacle (blue) with outer and inner cells.
    (b) Swept volume of robot over one angle interval (magnified).  (c)
    Approximate free space for this interval with close features in
    red.} \label{f-drone}
\end{figure}

\subsection{Free spaces}

The inputs are triangulated 3D free spaces for a polyhedral robot that rotates
around the $z$ axis and translates freely in the $xy$ plane relative to a
stationary polyhedron (Fig.~\ref{f-maze}), which we compute with our prior
algorithm \cite{sacks-butt-milenkovic17}.  The free space coordinates are
$(x,y,t)$ with $t$ a rational parameterization of the rotation angle.  The $t$
coordinates of the vertices are zeros of quartic polynomials whose coefficients
are polynomials in the coordinates of the input vertices.  Their algebraic
degree and precision are far higher than in the previous tests.

We test four robots with 4 to 42 triangles and four obstacles with 736 to 8068
triangles.  The 16 free spaces have median 125,000 and maximum 474,000
triangles, and median 530,000 and maximum 1,507,000 close features.
Modification reduces the number of close features to median 2,400 max 12,000,
and displaces median 17\% and maximum 26\% of the vertices with median and
maximum displacements of $0.07d$ and $1.31d$.  Expansion displaces median 1\%
and maximum 5\% of the vertices with median and maximum displacements of $0.8d$
and $8d$.  Optimization barely reduces the displacement and is slow, so we omit
it from the results.  The median and maximum running times per input triangle
are 0.001 and 0.002 seconds.  Table~\ref{t-simp} shows the results for the two
largest inputs.

\begin{figure}[htbp]
  \begin{tabular}{c}
    \includegraphics[scale=.2]{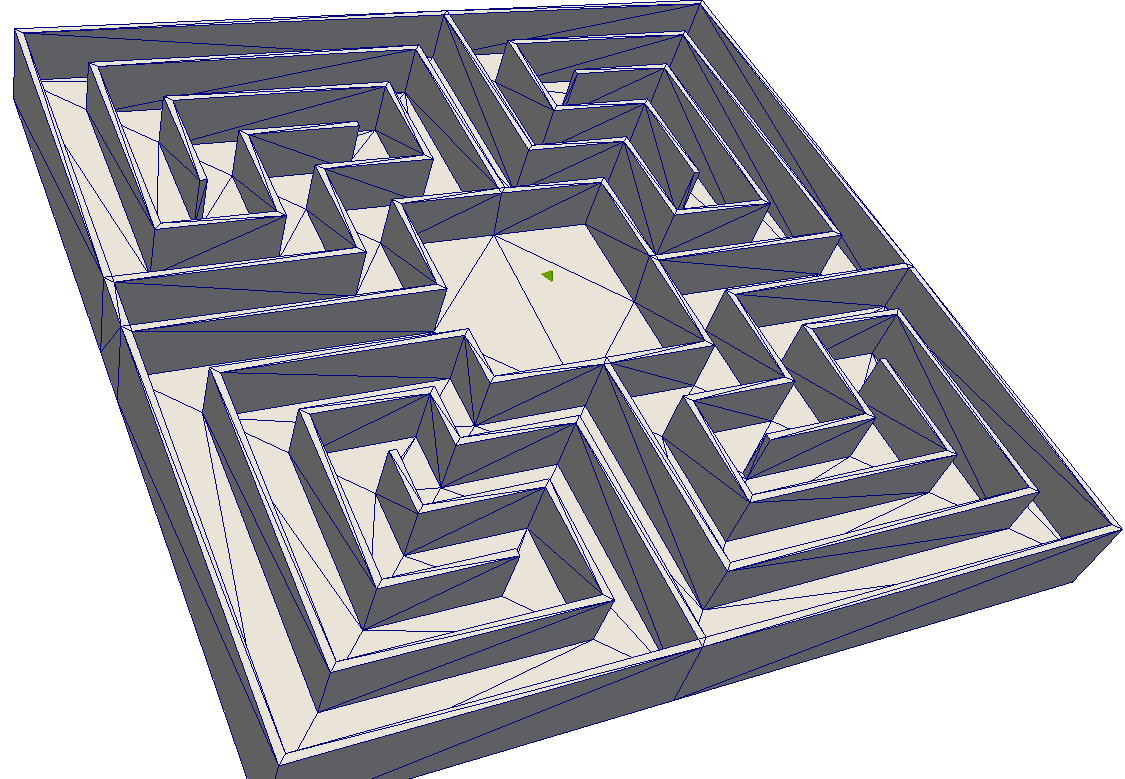}\\
    (a)\\
    \includegraphics[scale=.2]{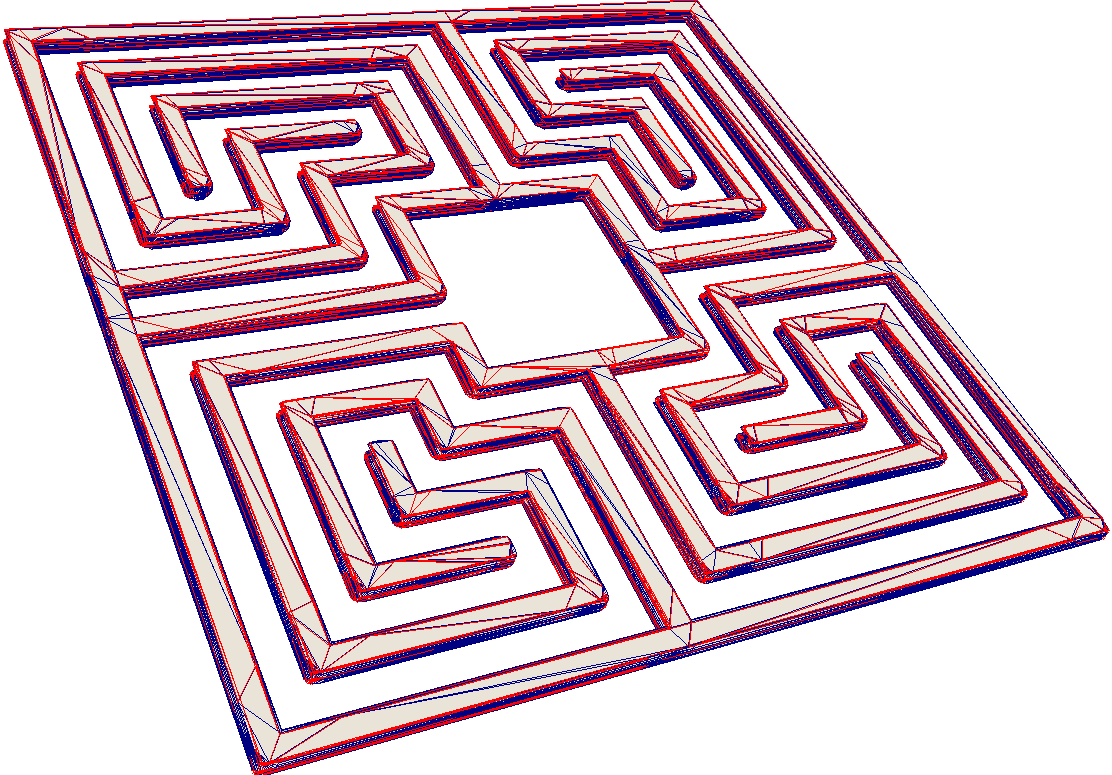}\\
    (b)
  \end{tabular}
  \caption{(a) Robot in maze; (b) configuration space with close features in
    red.}
  \label{f-maze}
\end{figure}

\begin{figure}[htbp]
  \centerline{\includegraphics[scale=.25]{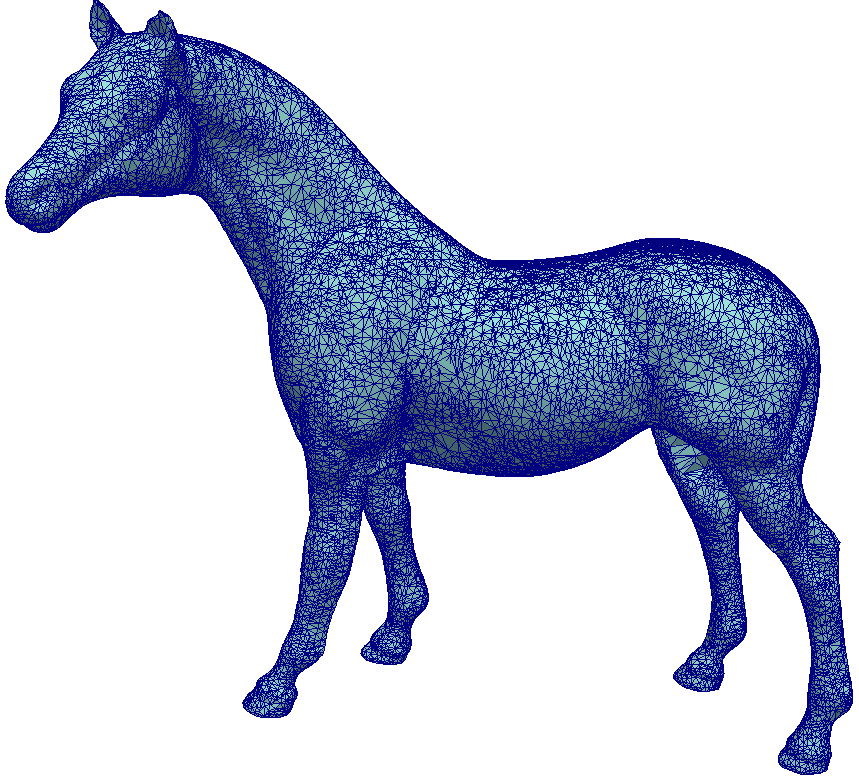}}
  \caption{Polyhedron whose tetrahedral mesh has 150,000 triangles.}
\label{f-horse}
\end{figure}

\subsection{Tetrahedral meshes}

The inputs are eight Delaunay tetrahedralizations with 3,000 to 150,000
triangles, and median 20 maximum 500 close features (Fig.~\ref{f-horse}).
Modification does nothing because no edge is incident on two triangles.
Expansion displaces median 4.9\% and maximum 30\% of the vertices with median
and maximum displacements of $0.75d$ and $3.96d$.  Optimization does not help.
The median and maximum running times per mesh triangle are $9\times10^{-5}$ and
$5\times10^{-4}$ seconds.  Table~\ref{t-simp} shows the results for the two
largest inputs.

\section{Conclusion}\label{s-con}

The feature separation algorithm performs well on a wide range of inputs.  The
isosurfaces and the tetrahedral meshes have few close features and 53 bit
precision.  The Minkowski sums have median 80 close features and 700 bit
precision.  The sweeps have median 1600 close features and 2100 bit precision.
The free spaces have median 530,000 close features and even higher precision.
In every case, the number of displaced vertices is proportional to the number of
close features, the median vertex displacement is bounded by $d$, and the
running time is proportional to the number of triangles plus close features with
a mild dependency on the input precision.

This performance corresponds to the assumption that every pair of close features
is far from every other pair.  The complexity of modification is linear in the
number of short edges and skinny triangles because mesh edits do not create
close features.  The error is bounded by $d$ because a vertex is displaced at
most once.  The expansion LP achieves the optimal displacement with
$\Delta=d$ because the truncation error is negligible.  Thus, the number of
steps is constant.  Optimization converges in a few steps because the initial
displacement is close to $d$.  A more realistic assumption is that every pair of
close features is close to a bounded number of close features.  A complexity
analysis under this assumption is a topic for future work.

We tested the benefit of modification by performing feature separation via
expansion alone.  The median vertex displacement increases slightly, the mesh
size increases proportionally to the number of short edges, and the running time
increases sharply.  To understand the small change in displacement, consider two
vertices that are $e$ apart.  Modification displaces them by $e/2$, versus
$(d-e)/2$ for expansion.  Assuming that $e$ is uniform on $[0,d]$, the average
displacement is the same.  The running time increases because mesh edits take
constant time, whereas expansion is quadratic in the number of close features.

We compared our linear distance function, which picks the $l$ and $m$ values, to
the standard function in which $l$ and $m$ are constants.  If $l$ and $m$ are
not bounded, the two formulations are equivalent.  (We omit the lengthy proof.)
The unbounded version is much slower that our version.  If $l$ and $m$ are
bounded by 1 instead of 0.001, little changes.  If they are set to zero, the
running time drops, the error grows, and in rare cases expansion does not
converge.

We conclude that modification plus expansion achieve feature separation quickly
and with a small displacement.  Optimization rarely achieves even a factor of
two reduction in displacement and sometimes is expensive.  Perhaps the reduction
could be increased by global optimization of the displacement subject to the
separation constraints, but this seems impractical because the constraints have
high dimension, are piecewise polynomial, and are not convex.

The results on tetrahedral meshes suggest that feature separation can play a
role in mesh improvement.  It simultaneously improves the close features by
displacing multiple vertices in a locally optimal manner, whereas prior work
optimizes one vertex at a time.  The tests show that the simultaneous approach
is fast and effective.  It readily extends to control other aspects of the mesh,
such as triangle normals.  Applying the approach to other mesh improvement tasks
is a topic for future work.

\section*{Acknowledgments}

Sacks is supported by NSF grant CCF-1524455.  Milen\-kovic is supported by NSF
grant CCF-1526335.

\bibliographystyle{elsarticle-num}
\bibliography{bib,bib2}

\newcommand{\noopsort}[1]{}
\begin{thebibliography}{10}
\expandafter\ifx\csname url\endcsname\relax
  \def\url#1{\texttt{#1}}\fi
\expandafter\ifx\csname urlprefix\endcsname\relax\def\urlprefix{URL }\fi
\expandafter\ifx\csname href\endcsname\relax
  \def\href#1#2{#2} \def\path#1{#1}\fi

\bibitem{ecg}
Exact computational geometry, http://cs.nyu.edu/exact.

\bibitem{sacks-milenkovic13b}
E.~Sacks, V.~Milenkovic, Robust cascading of operations on polyhedra,
  Computer-Aided Design 46 (2014) 216--220.

\bibitem{kyung-sacks-milenkovic15}
M.-H. Kyung, E.~Sacks, V.~Milenkovic, Robust polyhedral {Minkowski} sums with
  {GPU} implementation, Computer-Aided Design 67–68 (2015) 48--57.

\bibitem{cgal}
\textsc{Cgal}, {C}omputational {G}eometry {A}lgorithms {L}ibrary,
  http://www.cgal.org.

\bibitem{hp-isr-02}
D.~Halperin, E.~Packer, Iterated snap rounding, Computational Geometry: Theory
  and Applications 23~(2) (2002) 209--222.

\bibitem{m-spr97}
V.~Milenkovic, Shortest path geometric rounding, Algorithmica 27~(1) (2000)
  57--86.

\bibitem{goodrich97snap}
M.~T. Goodrich, L.~J. Guibas, J.~Hershberger, P.~J. Tanenbaum, Snap rounding
  line segments efficiently in two and three dimensions, in: Symposium on
  Computational Geometry, 1997, pp. 284--293.

\bibitem{fortune99}
S.~Fortune, Vertex-rounding a three-dimensional polyhedral subdivision,
  Discrete and Computational Geometry 22 (1999) 593--618.

\bibitem{fortune97}
S.~Fortune, Polyhedral modelling with multiprecision integer arithmetic,
  Computer-Aided Design 29~(2) (1997) 123--133.

\bibitem{zhou2016}
Q.~Zhou, E.~Grinspun, D.~Zorin, A.~Jacobson, Mesh arrangements for solid
  geometry, ACM Transactions on Graphics 35~(4) (2016) 39:1--39:15.

\bibitem{cignoni98}
P.~Cignoni, C.~Montani, R.~Scopigno, A comparison of mesh simplification
  algorithms, Computers and Graphics 22~(1) (1998) 37--54.

\bibitem{freitag2000}
L.~A. Freitag, P.~Plassmann, Local optimization-based simplicial mesh
  untangling and improvement, International Journal of Numerical Methods in
  Engineering 49 (2000) 109--125.

\bibitem{knupp2001}
P.~M. Knupp, Hexahedral and tetrahedral mesh untangling, Engineering with
  Computers 17~(3) (2001) 261--268.

\bibitem{cheng12}
S.-W. Cheng, T.~K. Dey, J.~Shewchuk, Delaunay Mesh Generation, Chapman and
  Hall, 2012.

\bibitem{sacks-butt-milenkovic17}
E.~Sacks, N.~Butt, V.~Milenkovic, Robust free space construction for a
  polyhedron with planar motion, Computer-Aided Design 90C (2017) 18--26.

\end{thebibliography}

\section*{Vitae}

Sacks is a professor in the Department of Computer Science of Purdue University.
His research interests are Computational Geometry, Computer Graphics, and
Mechanical Design.  Butt is a lead principal engineer at Autodesk.  He received
his PhD from Purdue in 2017 under the supervision of Sacks.  Milenkovic is a
professor in the Department of Computer Science of the University of Miami.  His
interests are Computational Geometry, Packing and Nesting, Graphics, and
Visualization.

\end{document}